\documentclass{article}
\usepackage[utf8]{inputenc}
\usepackage{latexsym}
\usepackage[normalem]{ulem}
\usepackage{hyperref}

\usepackage[outline]{contour}
\usepackage{comment}
\usepackage{graphicx}
\usepackage{color}
\usepackage{mathtools}
\usepackage{mathtools,halloweenmath}
\usepackage{amsmath, amsthm, amssymb}
\usepackage{enumerate}
\usepackage{float}
\usepackage{url}
\usepackage{stackrel}
\usepackage{mathrsfs,dsfont}
\usepackage{caption}
\usepackage{subcaption}
\usepackage{wrapfig}
\usepackage{authblk}
\usepackage{bbm}
\usepackage{bm}
\usepackage{epigraph}
\usepackage{tikz-qtree}
\usepackage{tikz}
\usepackage{enumitem}
\usepackage{url}
\usepackage{appendix}
\usepackage{caption}
\usepackage{subcaption}

\def\NN{\mathcal{N}}

\usepackage{tikz}      
\usetikzlibrary{shapes,automata,positioning,arrows,fit} 
\usetikzlibrary{decorations.pathmorphing}

\newtheorem{theorem}{Theorem}[section]

\newtheorem{proposition}[theorem]{Proposition}
\newtheorem{lemma}[theorem]{Lemma}
\newtheorem{definition}[theorem]{Definition}
\newtheorem{example}[theorem]{Example}

\definecolor{ratecnst}{RGB}{172,172,172}
\newcommand{\been}{\begin{enumerate}}
\newcommand{\enen}{\end{enumerate}}
\newcommand{\beit}{\begin{itemize}}
\newcommand{\enit}{\end{itemize}}
\def\bal#1\eal{\begin{align}#1\end{align}}
\def\bal*#1\eal*{\begin{align}#1\end{align}}

\def\bedf#1\endf{\begin{definition}#1\end{definition}}

\def\R{\mathbb R}

\usepackage{algorithm, algpseudocode}

\makeatletter
\newenvironment{breakablealgorithm}
  {
     \refstepcounter{algorithm}
     \hrule height.8pt depth0pt \kern2pt
     \renewcommand{\caption}[2][\relax]{
       {\raggedright\textbf{\fname@algorithm~\thealgorithm} ##2\par}%
       \ifx\relax##1\relax 
         \addcontentsline{loa}{algorithm}{\protect\numberline{\thealgorithm}##2}%
       \else 
         \addcontentsline{loa}{algorithm}{\protect\numberline{\thealgorithm}##1}%
       \fi
       \kern2pt\hrule\kern2pt
     }
  }{
     \kern2pt\hrule\relax
  }

\def\WR0{\textrm{WR}_0}

\def\supp{\textrm{supp}}

\def\cone{\textrm{Cone}}

\makeatletter
\newcommand{\uset}[3][0ex]{%
  \mathrel{\mathop{#3}\limits^{
    \vbox to#1{\kern+4.5\ex@
    \hbox{#2}\vss}}}}
\makeatother

\begin{document}

\title{Weakly reversible deficiency zero realizations of  reaction networks}
\author[1]{Neal Buxton}
\author[2]{Gheorghe Craciun}
\author[3]{Abhishek Deshpande}
\author[1]{Casian Pantea}
\affil[1]{\small Department of Mathematics, West Virginia University}
\affil[2]{\small Department of Mathematics and Department of Biomolecular Chemistry, \protect \\ University of Wisconsin-Madison}
\affil[3]{Center for Computational Natural Sciences and Bioinformatics, \protect \\
 International Institute of Information Technology Hyderabad}

\date{}

\maketitle

\begin{abstract}
\noindent
We prove that if a given reaction network $\NN$ has a weakly reversible deficiency zero realization for all choice of rate constants, then there exists a {\em unique} weakly reversible deficiency zero network $\NN'$ such that $\NN$ is realizable by $\NN'$.  Additionally, we propose an algorithm to find this weakly reversible deficiency zero network $\NN'$ when it exists.
\end{abstract}

\section{Introduction}
Polynomial dynamical systems are ubiquitous in many types of models in applied mathematics. These are usually complicated high-dimensional systems that depend on many parameters. Their analysis often relies on numerical methods and is sensitive to changes in parameter values. On the other hand, theoretical results that are robust with respect to parameter changes are rare. One prominent example comes from reaction network theory~\cite{craciun2005multiple,craciun2006multiple,wilhelm1995smallest, wilhelm1996mathematical,erdi1989mathematical,deshpande2014autocatalysis,craciun2022homeostasis,craciun2022autocatalytic}, the field of applied mathematics concerned with modeling chemical reaction networks, and particularly focused on the question: 

\medskip

{\em What dynamical properties of a chemical reaction network follow from the structure of the network alone, and are independent of parameter values in the model?}

\medskip

An important result in reaction network theory concerns the class of weakly reversible, deficiency-zero networks ($\WR0$ networks), which stems from thermodynamic considerations~\cite{horn1972general}. It was shown that $\WR0$ networks exhibit stable dynamics regardless of the values of parameters in the model. Specifically, it is known that $\WR0$ networks possess an asymptotically stable equilibrium within every invariant polyhedron ~\cite{horn1972general, horn1972necessary, feinberg1987chemical}. The Global Attractor Conjecture \cite{CraciunDickensteinShiuSturmfels2009, feinberg2019foundations} asserts that in fact this positive equilibrium is {\em globally} stable. 
The Global Attractor Conjecture is a rare theoretical statement about global stability of a large class of polynomial ODE systems,
and is relevant in other fields of applied mathematics. For example, the positive equilibrium in $\WR0$ networks called the Birch point plays an important role in algebraic statistics \cite{CraciunDickensteinShiuSturmfels2009}. Moreover, the Global Attractor Conjecture can also be viewed as a discrete version of the Boltzmann equation for ideal gases~\cite{craciun2021reaction}.

The Global Attractor Conjecture has been the object of much study in recent years, with several important partial results \cite{pantea2012persistence,craciun2013persistence, craciun2011persistence,  anderson2011proof, anderson2020classes, gopalkrishnan2014geometric}. A proof strategy for the general result has been proposed in~\cite{craciun2019polynomial} using the idea of building invariant regions for toric differential inclusions~\cite{ding2021minimal,ding2022minimal,craciun2020endotactic}. 

It is known that different networks may give rise to (i.e., ``realize") the same ODE system. This property of {\em dynamical equivalence} has been explored algorithmically~\cite{craciun2024weakly,craciun2023weaklysingle,craciun2020efficient,craciun2024connectivity,craciun2023lower,craciun2024dimension,kothari2024endotactic} and in the analysis of inclusions between the sets of dynamical systems generated by two different networks  $\NN$ and $\NN'$~\cite{anderson2020classes,kothari2024realizations}. 

In particular, if the network $\NN'$ is $\WR0$, and any dynamical system generated by $\NN$ can also be generated by $\NN'$, then the network $\NN$ enjoys the same stability properties as $\NN'$, without necessarily being a $\WR0$ network itself. In this work we aim to identify these ``disguised" $\WR0$ networks, i.e. networks that generate the same dynamical systems as $\WR0$ networks. Our starting point is a recent algorithm \cite{craciunalgorithm} that identifies such networks for {\em fixed} parameter values in the model. Our first result is an algorithm where the conclusion holds {\em for all} parameter values.   

The paper is structured as follows. Section \ref{sec:prelim} introduces the terminology of reaction networks and a background on convex cones, including relevant results on $\WR0$ networks and network realizability. In section \ref{sec:wr0} we discuss in detail the recent algorithm for $\WR0$ realization for {\em fixed} parameters \cite{craciunalgorithm}. Our new result on $\WR0$ realizations  {\em for all} parameter values and its proof are contained in Section \ref{sec:mainres}. Software implementation is discussed in Section \ref{sec:soft}. In Section \ref{sec:future}, we summarize our findings and list avenues for future work.  

\section{Preliminaries}\label{sec:prelim}
\subsection{A few remarks on polyhedral cones}
We start by recalling standard results on polyhedral cones \cite{ziegler2012lectures,rockafellar1970convex,fulton1993introduction,boyd2004convex}.
A polyhedral cone is a finitely generated convex cone: for $v_1,\ldots, v_k\in {\mathbb R}^n$
$$Cone(v_1,\ldots,v_k)=\{\alpha_1v_1+\ldots+\alpha_kv_k|\alpha_1,\ldots,\alpha_k\ge 0\}.$$
Its (relative) interior is  
$$Cone(v_1,\ldots,v_k)=\{\alpha_1v_1+\ldots+\alpha_kv_k|\alpha_1,\ldots,\alpha_k> 0\}.$$

The faces of $\cal C$ are also polyhedral cones. If $\cal C$ is a polyhedral cone and $S\subset \cal C$ then the {\em face generated by S} is the minimal face of $\cal C$ which includes $S$, i.e. the only face of $\cal C$ that contains $S$ in its interior:
$$\Phi(S)=  \bigcap_{F \text{\,is a face of }{\cal C} \text{ containing }S}F.$$

A polyhedral cone $\cal C$ is {\em pointed} if it does not contain any lines. If $\cal C$ is pointed then a subset $\{c_1,\ldots, c_l\}$ of $\{v_1,\ldots, v_n\}$ is a {\em complete set of extreme rays} of $\cal C$ if ${\cal C}=Cone(c_1,\ldots, c_l)$ and $l$ is minimal with this property. Extreme rays are unique up to multiplication by positive scalars. 

A cone generated by linearly independent vectors $v_1,\ldots, v_k$ is called {\em simplicial}. Simplicial cones $Cone(v_1,\ldots v_k)$ are pointed and $v_1,\ldots, v_k$ form a complete set of extreme rays. Faces of $Cone(v_1,\ldots, v_k)$ are precisely $Cone(S)$ for $S\subseteq\{v_1,\ldots v_k\}$.

\subsection{Reaction networks}\label{sec:reaction_networks}

A {\em reaction} on a finite list of {\em species} $X=(X_1,\ldots, X_n)$ has the form 
$$y_1X_1+\ldots y_nX_n\to y'_1X_1+\ldots y'_nX_n,$$
or in short $y\cdot X\to y'\cdot X$, where the coefficients $y_i$ and $y'_i$ are non-negative integers. A {\em reaction network} is a finite collection of reactions on the same set of species.

With a fixed ordering of species, a reaction network $\NN$ can be identified with a directed graph $(V,E)$ with $V\subset {\mathbb R}^n$. Specifically reaction $y\cdot X\to y'\cdot X$ is represented by the directed edge $y\to y'$.  This graph is called a {\em Euclidean embedded graph}, or {\em E-graph}~\cite{craciun2019quasi,craciun2015toric,craciun2019polynomial,craciun2020endotactic,craciun2020efficient,craciun2024weakly,craciun2021uniqueness,deshpande2023source}. We take this identification as a working definition:

\begin{definition}[Reaction network]
A {\em reaction network} $\NN$ on $n$ species is a finite directed graph $(V,E)$ where $V \subset \mathbb{R}^n$ and every vertex in $V$ is incident to at least one edge. Vertices in $V$ are called {\em complexes} and  edges in $E$ are called {\em reactions}. A reaction $y\to y'$ has {\em source complex} $y$ and {\em product complex} $y'$. The connected components of $\NN=(V,E)$ are called {\em linkage classes}. Further, we say that a complex belongs to a linkage class if it is a vertex of that connected component. 
\end{definition}

By definition all complexes participate in some reaction; we may therefore identify a reaction network $\NN$ with its set of reactions. A {\em rate constant} $\kappa_{y\to y'}>0$ may be assigned to each reaction $y\to y'$. For notation purposes, we fix an ordering of reactions in $\NN$ and collect the rate constants in a vector $\kappa$. 
This yields the {\em mass-action system} $(\NN,\kappa)$ and its {\em mass-action ODE system~\cite{feinberg1979lectures,gunawardena2003chemical,guldberg1864studies,voit2015150,yu2018mathematical,adleman2014mathematics}}.

\begin{equation}\label{eq:ma}
\frac{dx}{dt}=f_{(\NN,\kappa)}(x)=\sum_{y\to y'\in\NN} \kappa_{y\to y'} x^{y}(y'-y)
\end{equation}
\noindent where $x^y$ denotes the monomial $x_1^{y_1}x_2^{y_2}\ldots x_n^{y_n}$.

Equation~\ref{eq:ma} is a polynomial ODE system on the {\em concentration vector} $x(t)=(x_1(t),\ldots, x_n(t))\in {\mathbb R}^n$ of $X=(X_1,\ldots, X_n)$. We  restrict our focus to non-negative solutions of (\ref{eq:ma}), i.e. non-negative initial conditions $x(0)$.

The {\em reaction vector} of reaction $y\to y'$ is defined to be the vector  $y'-y$. Reaction vectors span the {\em stoichiometric subspace} of $\NN$, denoted 
${\cal S}_{\NN}=\text{span}\{y'-y\,|\,y\to y'\in E\}.$ 
By integrating (\ref{eq:ma}) we see that  
$x(t)\in (x(0)+{\cal S}_{\NN})\cap {\mathbb R}_{\ge 0}^n$, in other words solutions $x(t)$ are constrained to translations of the stoichiometric subspace into the non-negative orthant, termed {\em stoichiometric compatibility classes}. 

\medskip

Grouping terms in $f_{(\NN,\kappa)}$, the vector of coefficients of $x^y$ is 
\begin{equation}\label{eq:netrv}
w_y=\sum_{y\to y'\in\NN}\kappa_{y\to y'}(y'-y)
\end{equation} 
referred to as the {\em net reaction vector} of source complex $y$ \cite{craciun2021uniqueness}. The net reaction vector of complex describes the relative interior (see below) of $Cone_{\NN}(y)$ when $\kappa$ is varied ($Cone_{\NN}(y)$ is defined in the same way but with $\kappa\ge 0$):
\begin{equation}\label{eq:varyk}
\left\{\sum_{y\to y'\in\NN}\kappa_{y\to y'}(y'-y)|\kappa>0\right\}=\text{int }(Cone_{\NN}(y))  
\end{equation}

Letting $y_1,\ldots, y_m$ denote the source complexes of $\NN$ with non-zero net reaction vectors $w_{y_1},\ldots,w_{y_m}$ we may write

\begin{equation}\label{eq:netst}
f_{({\NN},\kappa)}(x)=\sum_{j=1}^m w_{y_j}x^{y_j}=Wx^Y.
\end{equation}
The columns of $W$ are net reaction vectors  
$W=[w_{y_1}|\ldots|w_{y_m}]$ and the columns of $Y$ are the corresponding source complexes of $\NN$,
$Y=[y_1|\ldots |y_m].$ By notation  $x^{Y}=[x^{y_1},\ldots, x^{y_m}]^T.$ 

All columns of $Y$ are pairwise distinct, and all columns of $W$ are non-zero. Any polynomial ODE system on ${\mathbb R}^n$ is written uniquely as $\frac{dx}{dt}=Wx^Y$ for some matrices $W, Y\in{\mathbb R}^{n\times m}$
with these properties (here $m$ is the number of monomials in the ODE system). However, not all polynomial ODE systems are mass-action systems, i.e. there may not exist mass-action systems $(\NN,\kappa)$ for which 
$Wx^Y=f_{(\NN,\kappa)}$. Polynomials appearing on the right hand side of mass-action ODE systems are characterized by the Hungarian Lemma \cite{HungarianLemma}. 

\medskip

\subsection{Weakly reversible deficiency zero networks}\label{sec:wr0}

A reaction network must be {\em weakly reversible} to admit a thermodynamical equilibrium (so-called ``complex balanced" equilibrium) \cite{horn1972general}:

\begin{definition}[Weakly reversible networks]  
A reaction network $\NN$ is called {\em weakly reversible} if all its connected components are strongly connected (i.e. each reaction is a part of a directed cycle).
\end{definition}

The {\em deficiency} of a network is a non-negative integer computed from the structure of the network itself. It is a measure of how much the linear structure of the network is coupled with the nonlinear exponents in the ODE system: the smaller the deficiency, the more decoupling.

\begin{definition}[Deficiency]
The {\em deficiency} of reaction network $\NN$ is 
$$\delta=|V|- l- \dim({\cal S}_{\NN})$$
where $l$ is the number of linkage classes of $\NN$ and $|V|$ denotes the number of complexes of $\NN$. 
\end{definition}

Deficiency proved a useful tool in answering questions on 
reaction networks, including multistationarity \cite{conradi2007subnetwork}. It turns out that deficiency zero, weakly reversible networks (termed {\em $\WR0$ networks} here) have exceptionally stable dynamics: 

\begin{theorem}[Deficiency Zero Theorem \cite{feinberg1979lectures}] If $\NN$ is a $\WR0$ network, then for any rate constants $\kappa$ the mass-action system $(\NN,\kappa)$ has a unique positive equilibrium point in each compatibility class, and this equilibrium is asymptotically stable. 
\end{theorem}

In fact, it has long been conjectured that  equilibria of $\WR0$ systems are {\em globally} stable with respect to their compatibility classes. This is the so-called Global Attractor Conjecture \cite{CraciunDickensteinShiuSturmfels2009}, the focus of much work in recent years. The conjecture has been shown to be true when the stoichiometric subspace has dimension less than 3 \cite{CraciunDickensteinShiuSturmfels2009, pantea2012persistence, craciun2011persistence}. A proof of the general result was proposed recently in \cite{craciun2019polynomial}. 

\subsection{Realizations of mass-action systems}\label{sec:realizations}

It is possible for different mass-action systems to generate the same set of differential equations, as the next example shows.

\begin{example}\label{ex:1} Consider networks $\NN$ and $\NN_0$ in Figure \ref{fig:ex1}. 
Choosing all rate constants equal to 3 for $\NN$ and all rates equal to $1$ for ${\NN_0}$, the mass-action systems $(\NN,{\bf 3})$ and $(\NN_0,{\bf 1})$  generate the same ODE system

\begin{equation}\label{eq:ex1ode}
\begin{bmatrix}
\dot x\\
\dot y
\end{bmatrix}
=\begin{bmatrix}
1&0&-1\\
1&-1&0
\end{bmatrix}
\begin{bmatrix}
1\\
xy^2\\
x^2y
\end{bmatrix}
=Wx^Y.
\end{equation}

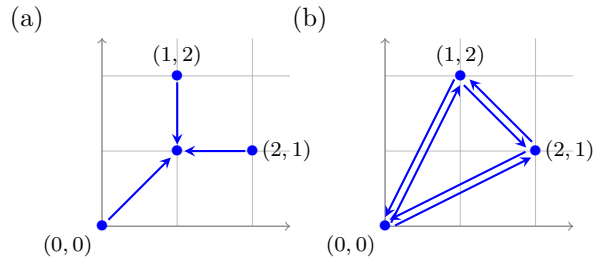
\begin{figure}[h!tbp]\label{fig:ex1}
	\centering
	\begin{subfigure}[b]{0.3\textwidth}
	\centering
		\begin{tikzpicture}
        \draw [step=1, gray!50!white, thin] (0,0) grid (2.5,2.5);
        \node at (-1,2.75) {(a)};
		\node at (0,-0.25) {};
            \draw [->, gray] (0,0)--(2.5,0);
            \draw [->, gray] (0,0)--(0,2.5);
            \node [inner sep=0pt, blue] (2x1y) at (2,1) {$\bullet$};
            \node [inner sep=0pt, blue] (1x2y) at (1,2) {$\bullet$};
            \node [inner sep=0pt, blue] (0) at (0,0) {$\bullet$};
            \node [inner sep=0pt, blue] (xy) at (1,1) {$\bullet$};
            \draw [-{stealth}, thick, blue] (0)--(xy) node [midway, left]{};
            \draw [-{stealth}, thick, blue] (2x1y)--(xy) node [midway, above]{}; 
            \draw [-{stealth}, thick, blue] (1x2y)--(xy) node [midway, left]{}; 
            \node at (0) [below left] {\footnotesize $(0,0)$};
            \node at (1x2y) [above] {\footnotesize $(1,2)$};
            \node at (2x1y) [right] {\footnotesize $(2,1)$};
		\end{tikzpicture}
    \end{subfigure}
    \begin{subfigure}[b]{0.3\textwidth}
	\centering
		\begin{tikzpicture}
        \draw [step=1, gray!50!white, thin] (0,0) grid (2.5,2.5);
        \node at (-1,2.75) {(b)};
		\node at (0,-0.25) {};
            \draw [->, gray] (0,0)--(2.5,0);
            \draw [->, gray] (0,0)--(0,2.5);
            \node [inner sep=0pt, blue] (2x1y) at (2,1) {$\bullet$};
            \node [inner sep=0pt, blue] (1x2y) at (1,2) {$\bullet$};
            \node [inner sep=0pt, blue] (0) at (0,0) {$\bullet$};
            \draw [-{stealth}, thick, blue, transform canvas={xshift=0.25ex, yshift=-0.25ex}] (0)--(1x2y) node [midway, left]{}; 
            \draw [-{stealth}, thick, blue, transform canvas={xshift=-0.25ex, yshift=0.25ex}] (1x2y)--(0) node [midway, left]{}; 
            \draw [-{stealth}, thick, blue, transform canvas={xshift=-0.25ex, yshift=0.25ex}] (2x1y)--(0) node [midway, above]{}; 
            \draw [-{stealth}, thick, blue, transform canvas={xshift=0.25ex, yshift=-0.25ex}] (0)--(2x1y) node [midway, above]{}; 
            \draw [-{stealth}, thick, blue, transform canvas={xshift=-0.25ex, yshift=-0.25ex}] (1x2y)--(2x1y) node [midway, left] {};
            \draw [-{stealth}, thick, blue, transform canvas={xshift=0.25ex, yshift=0.25ex}] (2x1y)--(1x2y) node [midway, left] {}; 
            \node at (0) [below left] {\footnotesize $(0,0)$};
            \node at (1x2y) [above] {\footnotesize $(1,2)$};
            \node at (2x1y) [right] {\footnotesize $(2,1)$};
		\end{tikzpicture}
    \end{subfigure}
\caption{Network $\NN$ in (a) and $\NN_0$ in (b) give rise to the same ODE system if we choose all rate constants equal to 3 for $\NN$ and all rates equal to $1$ for ${\NN_0}$.}\label{fig:ex1}
\end{figure}
\end{example}

Mass-action systems $(\NN,{\bf 3})$ and $(\NN_0,{\bf 1})$  are called {\em realizations of mass-action system (\ref{eq:ex1ode})}.  
\begin{definition}[Realizations of mass-action systems]
The mass-action system $(\NN',\kappa')$ is called a {\em realization} of the mass-action ODE system $\frac{dx}{dt}=Wx^Y$ if $f_{(\NN',\kappa')}=Wx^Y.$ The mass-action system $(\NN',\kappa')$ is called a realization of $(\NN,\kappa)$ if it is a realization of $\dot x =f_{(\NN,\kappa)}(x)$.

\end{definition}

We may replace the rate constants $\kappa={\bf 1}$ in example \ref{ex:1} by any positive numbers, and the conclusion remains: there exist rate constants $\kappa'$ such that $(\NN',\kappa')$ is a realization of $\dot x =f_{(\NN,\kappa)}(x)$. In other words all possible mass-action dynamics generated by $\NN$ can also be seen as mass-action dynamics generated by $\NN'$:

\begin{definition}[Dynamically realizable networks]\label{def:real}
Let $\NN$, $\NN'$ be two reaction networks on the same list of species. $\NN$ is called {\em  realizable} by  $\NN'$ if for any choice of rate constants $\kappa$ 
for $\NN$ there exist rate constants $\kappa'$ for $\NN'$ such that 
$f_\NN(\kappa,x)=f_{\NN'}(\kappa',x)$  for all $x\in\mathbb{R}^n_{>0}$. We call $\NN'$ a {\em realization} of $\NN$. 
\end{definition}

A network $\NN$ that is realizable by $\NN'$ can be also rephrased as the dynamics of $\NN'$ contains the dynamics of $\NN$ in the language of~\cite{anderson2020classes}. Deshpande et.al~\cite{kothari2024realizations} refer to it as $\NN \sqsubseteq \NN'$.

In other words $\NN'$ is a realization of $\NN$ if for any rate constants $\kappa$ of $\NN$ there exist rate constants $\kappa'$ for $\NN'$ such that the mass-action system   $(\NN',\kappa')$ is a realization of the mass-action ODE system $\dot x=f_{(\NN,\kappa)}(x)$. 

Definition~\ref{def:real} is closely related to the notions of dynamical equivalence. Dynamical equivalence has also been called \emph{coufoundability} by  Craciun and Pantea~\cite{craciun2008identifiability} and \emph{macroequivalence} by Horn and Jackson~\cite{horn1972general}. Dynamical realizability has the following simple characterization  
\cite{craciun2008identifiability, anderson2020classes}

\begin{theorem}\label{thm:dynIncl}
Let $\NN$, $\NN'$ be two reaction networks on the same list of species. Then $\NN$ is {\em dynamically realizable} by  $\NN'$ if and only if for any source complex $y$ of $\NN$, $y$ is also a source complex of $\NN'$ and $Cone_{\NN}(y)\subseteq Cone_{\NN'}(y)$.
\end{theorem}

\section{$\WR0$ realizations of mass-action systems}\label{se:wr0real}

 Mass-action ODE systems may have multiple realizations, for some of which results from  reaction network theory may apply. In this paper we are interested in realizations $(\NN_0,\kappa_0)$ where $\NN_0$ is a $\WR0$ network. If such realizations exist, then the results in Section \ref{sec:wr0} allow powerful conclusions to be drawn on the dynamics of $(\NN,\kappa)$ without $\NN$ necessarily being a $\WR0$ network. This motivates the following definition.

\begin{definition}[$\WR0$ realization of mass-action systems] We say that the mass-action ODE system $\frac{dx}{dt}=Wx^Y$ on ${\mathbb R^n}$  {\em has $\WR0$ realization} if there exists a realization $(\NN_0,\kappa_0)$ of $(\NN,\kappa)$ where $\NN_0$ is a $\WR0$ network. We say that the mass-action system $(\NN,\kappa)$ has $\WR0$ realization if its corresponding ODE system does. 
\end{definition} 

\begin{example}\label{ex:2}

Consider a network $\NN$ in Figure \ref{fig:ex2}(a). This is not a weakly reversible network, but one may ask the question: After fixing rate constants $\kappa$, is there a $\WR0$ system that has the same ODE system?

\begin{figure}[h!tbp]
	\centering
    \begin{subfigure}[b]{0.3\textwidth}
	   \centering
		\begin{tikzpicture}[scale=2]
        \draw [step=1, gray!50!white, thin] (0,0) grid (1.5,1.5);
        \node at (-1,1.5) {};
        \node at (-0.5,1.5) {(a)};
            \draw [->, gray] (0,0)--(1.5,0);
            \draw [->, gray] (0,0)--(0,1.5);
            \node [inner sep=0pt, blue] (0) at (0,0) {$\bullet$};
            \node [inner sep=0pt, blue] (0x1y) at (0,1) {$\bullet$};
            \node [inner sep=0pt, blue] (1x0y) at (1,0) {$\bullet$};
            \node [inner sep=0pt, blue] (1x1y) at (1,1) {$\bullet$};
            \draw [-{stealth}, thick, blue] (0)--(1x1y) node [midway, left]{$\kappa_1$};
            \draw [-{stealth}, thick, blue] (1x1y)--(1x0y) node [midway, right] {$\kappa_2$};
            \draw [-{stealth}, thick, blue] (1x1y)--(0x1y) node [midway, above] {$\kappa_3$};
            \node at (0) [below left] {\footnotesize $(0,0)$};
            \node at (1x1y) [above right] {\footnotesize $(1,1)$};
            \node at (0x1y) [left] {\footnotesize $(0,1)$};
            \node at (1x0y) [below] {\footnotesize $(1,0)$};
		\end{tikzpicture}
    \end{subfigure}
    \begin{subfigure}[b]{0.3\textwidth}
	   \centering
		\begin{tikzpicture}[scale=2]
        \draw [step=1, gray!50!white, thin] (0,0) grid (1.5,1.5);
        \node at (-1,1.5) {};
        \node at (-0.5,1.5) {(b)};
            \draw [->, gray] (0,0)--(1.5,0);
            \draw [->, gray] (0,0)--(0,1.5);
            \node [inner sep=0pt, blue] (0) at (0,0) {$\bullet$};
            \node [inner sep=0pt, blue] (1x1y) at (1,1) {$\bullet$};
            \draw [-{stealth}, thick, blue, transform canvas={xshift=0.25ex, yshift=-0.25ex}] (0)--(1x1y) node [midway, right]{$k_1$}; 
            \draw [-{stealth}, thick, blue, transform canvas={xshift=-0.25ex, yshift=0.25ex}] (1x1y)--(0) node [midway, left]{$k_2$}; 
            \node at (0) [below left] {\footnotesize $(0,0)$};
            \node at (1x1y) [above right] {\footnotesize $(1,1)$};
		\end{tikzpicture}
    \end{subfigure}
    \caption{Mass-action system $\NN$ in (a) has $\WR0$ realization (b) when $\kappa_2=\kappa_3$}\label{fig:ex2}
\end{figure}
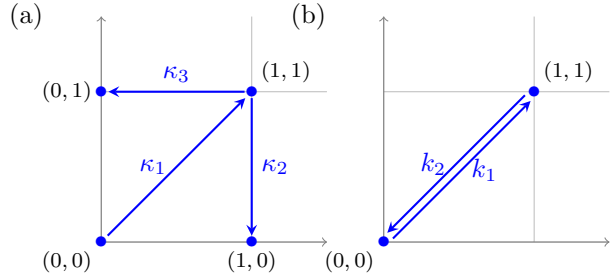

We see that if (and in fact only if) $\kappa_2=\kappa_3$, then such a realization exists: the network $\NN_0$ in Figure \ref{fig:ex2}(b) with rate constants 
$(k_1,k_2)=(\kappa_1, \kappa_2+\kappa_3)$
generate the same ODE system 

\begin{equation}\label{eq:ex2}
\begin{bmatrix}
\dot x\\
\dot y
\end{bmatrix}=
\begin{bmatrix}\kappa_1&-\kappa_3\\
\kappa_1&-\kappa_2\end{bmatrix}
\begin{bmatrix}1\\xy
\end{bmatrix}
\end{equation}


\begin{example}\label{ex:3}
Consider network $\NN$ in Figure \ref{fig:ex3}(a). Again, this is not a weakly reversible network.    

\begin{figure}[h!tbp]
	\centering
    \begin{subfigure}[b]{0.3\textwidth}
    \centering
		\begin{tikzpicture}[scale=2]
        \draw [step=1, gray!50!white, thin] (0,0) grid (1.5,1.5);
        \node at (-1,1.5) {};
        \node at (-0.5,1.5) {(a)};
            \draw [->, gray] (0,0)--(1.5,0);
            \draw [->, gray] (0,0)--(0,1.5);
            \node [inner sep=0pt, blue] (0) at (0,0) {$\bullet$};
            \node [inner sep=0pt, blue] (1x0y) at (1,0) {$\bullet$};
            \node [inner sep=0pt, blue] (1x1y) at (1,1) {$\bullet$};
            \node [inner sep=0pt, blue] (0x1y) at (0,1) {$\bullet$};
            \draw [-{stealth}, thick, blue] (0)--(1x0y) node [midway, below] {$\kappa_1$};
            \draw [-{stealth}, thick, blue] (1x0y)--(1x1y) node [midway, right] {$\kappa_2$}; 
            \draw [-{stealth}, thick, blue] (1x1y)--(0) node [midway, left] {$\kappa_3$}; 
            \draw [-{stealth}, thick, blue] (0)--(0x1y) node [midway, left] {$\kappa_4$};
            \node at (0) [below left] {\footnotesize $(0,0)$};
            \node at (0x1y) [above] {\footnotesize $(0,1)$};
            \node at (1x1y) [above right] {\footnotesize $(1,1)$};
            \node at (1x0y) [below] {\footnotesize $(1,0)$};
		\end{tikzpicture}
    \end{subfigure}
    \begin{subfigure}[b]{0.3\textwidth}
    \centering
		\begin{tikzpicture}[scale=2]
        \draw [step=1, gray!50!white, thin] (0,0) grid (1.5,1.5);
        \node at (-1,1.5) {};
        \node at (-0.5,1.5) {(b)};
            \draw [->, gray] (0,0)--(1.5,0);
            \draw [->, gray] (0,0)--(0,1.5);
            \node [inner sep=0pt, blue] (0) at (0,0) {$\bullet$};
            \node [inner sep=0pt, blue] (1x0y) at (1,0) {$\bullet$};
            \node [inner sep=0pt, blue] (1x1y) at (1,1) {$\bullet$};
            \draw [-{stealth}, thick, blue] (0)--(1x0y) node [midway, below]{$k_1$};
            \draw [-{stealth}, thick, blue] (1x0y)--(1x1y) node [midway, right]{$k_2$}; 
            \draw [-{stealth}, thick, blue, transform canvas={xshift=-0.25ex, yshift=0.25ex}] (1x1y)--(0) node [midway, left] {$k_3$}; 
            \draw [-{stealth}, thick, blue, transform canvas={xshift=0.25ex, yshift=-0.25ex}] (0)--(1x1y) node [midway, right] {$k_4$};
            \node at (0) [below left] {\footnotesize $(0,0)$};
            \node at (1x1y) [above right] {\footnotesize $(1,1)$};
            \node at (1x0y) [below] {\footnotesize $(1,0)$};
		\end{tikzpicture}
    \end{subfigure}
     \caption{Mass-action system $\NN$ in (a) has $\WR0$ realization (b) when $\kappa_1>\kappa_4$}\label{fig:ex3}
\end{figure}
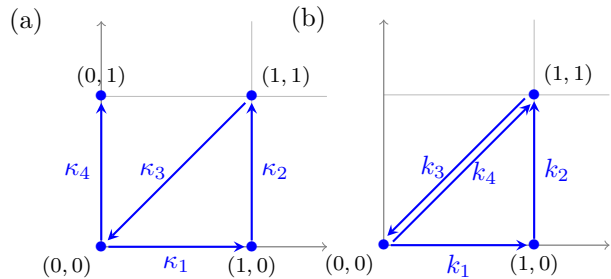

However, when $k_1>k_4$, its corresponding mass-action ODE system
\begin{equation}\label{eq:ex3}
\begin{bmatrix}\dot x\\\dot y\end{bmatrix}
=\begin{bmatrix}k_1&0&-k_3\\k_4&k2&-k_3\end{bmatrix}\begin{bmatrix}1\\x\\xy\end{bmatrix}
\end{equation}

is the same as that of the $\WR0$ mass-action system in Figure \ref{fig:ex3}(b) with $(k_1,k_2,k_3,k_4)=(\kappa_1-\kappa_4, \kappa_2,\kappa_3,\kappa_4)$.
\end{example}

\end{example}


Examples~\ref{ex:1} and \ref{ex:2} illustrate the usefulness of $\WR0$ realizations: based on their $\WR0$ realization, the ODE systems (\ref{eq:ex2}) and (\ref{eq:ex3}) have a unique positive steady state, which is globally stable with respect to ${\mathbb R}_{>0}^n$~\cite{craciun2011persistence,craciun2013persistence,craciun2015toric,gopalkrishnan2014geometric,anderson2011proof} whenever the corresponding conditions on the rate constants $\kappa$ are satisfied. These conclusions may be difficult to reach by inspection of the ODE system or network $\NN$ alone. All of these motivate the following: 

\noindent {\bf  Question 1} Does a given mass-action system $(\NN,\kappa)$ have $\WR0$ realization? \\

\noindent This question was answered recently in \cite{WR0alg}. The answer is the following algorithm, which is both the motivation and the main technical tool for our paper. We use the following notation for the support of a vector $v\in{\mathbb R}^n$: 
$supp(v)=\{i|1\le i\le n \text{ and } v_i\neq 0$\}.\\
\medskip

\begin{breakablealgorithm}\label{alg:polly}
\caption{($\WR0$ realizations of mass-action systems)}

\medskip

\noindent \textbf{Input:} 
Matrices $Y=(y_1,\ldots,y_m)$ and $W=(w_1,\ldots,w_m)$ that define a mass-action system written using net reaction vectors $\frac{dx}{dt}=\sum_{i=1}^mx^{y_i}w_i$. 

\medskip

\noindent \textbf{Output:} Either ``Success", and return the unique $\WR0$ realization $(\NN,\kappa)$ or ``Fail", and  print that such a realization does not exist.

\medskip

\begin{algorithmic}[1] 

\State initialize $(\NN_0,\kappa_0)= \emptyset$

\State find a minimal set of extreme rays $\{c_1,c_2,\ldots,c_l\}$ of the pointed convex cone $\ker(W)\cap\R_{\geq 0}^m$.

\If{$\supp(c_1),\supp(c_2),\ldots,\supp(c_l)$ do not form a partition of $\{1,2,\ldots,m\}$} 

\State print: Fail. $\WR0$ realization does not exist.  {\bf exit}.

\Else

define $V_p:=\supp(c_p)$ for $p=1,2,\ldots,l$.

\For{$p=1,2,\ldots,l$} 

\If{the vectors $\{y_i:i\in V_p\}$ are not affinely independent}

\State print: Fail. $\WR0$ realization does not exist.  {\bf exit}.

\Else

\For{$i\in V_p$}

\If{$w_i\not\in\cone\{y_j-y_i:j\in V_p\}$}

\State print: Fail. $\WR0$ realization does not exist.  {\bf exit}.

\Else

\State decompose $w_i = \displaystyle\sum_{j\in V_p,j\ne i}k_{ij}(y_j-y_i)$ with $k_{ij}\ge 0$.

\For {$j\in V_p, j\neq i$}

\If{$k_{ij}>0$}

\State add reaction $y_i\to y_j$ with rate constant $k_{ij}$ to mass-action system $(\NN_0,\kappa_0)$

\EndIf

\EndFor

\EndIf

\EndFor

\EndIf

\EndFor

\EndIf

\State print: Success. $\WR0$ realization exists. Return $(\NN_0,\kappa_0)$.

\end{algorithmic}

\end{breakablealgorithm}

\bigskip

\medskip

Algorithm \ref{alg:polly} combines in a streamlined procedure a few key technical facts on $\WR0$ mass-action systems, including the following, which explains the {\bf exit} lines in Algorithm
\ref{alg:polly}.

If {\bf exit} lines are avoided then Algorithm \ref{alg:polly} terminates by returning the $\WR0$ realization $(\NN_0,\kappa_0)$.

%



We end this section by recalling the following propositions from previous works.

\begin{proposition}\label{prop:2}\cite[Lemma 3.1]{craciunalgorithm}
Let $(\NN,\kappa)$ be a $\WR0$ mass-action system with corresponding mass-action ODE system $\dot x=Wx^Y,$ $W\in{\mathbb R}^{n\times m}$. Then $\ker W$ has a basis of non-negative vectors $c_1,\ldots, c_l$ whose supports partition $\{1,\ldots,m\}$ (each $c_i$ is supported on a linkage class of $\NN$). The vectors $c_1,\ldots, c_l$ form a complete set of extreme rays of the cone $\ker W\cap{\mathbb R_{\ge0}}$.
\end{proposition}

\begin{proposition}\label{prop:1}\cite{feinberg2019foundations,horn1972general,horn1972necessary} 
Let $\NN$ be a deficiency zero reaction network. Then the complexes of any linkage class of $\cal N$ are affinely independent.
\end{proposition}

\begin{theorem}\cite[Theorem 3.11]{craciun2021uniqueness}
A mass-action system $\dot x = Wx^Y$ has at most one $\WR0$ realization.
\end{theorem}

This implies that $\WR0$ realizations are unique, and we may refer to {\em the} $\WR0$ realization of the mass-action system if it exists.

\section{$\WR0$ realizations of mass-action networks}\label{sec:mainres}

Suppose we fix a reaction network $\NN$. A very interesting question is to describe a general algorithm finding the set of rate constants $\kappa$ of a mass-action system $(\NN,\kappa)$ whose ODE system has $\WR0$ realization, like we did in Examples \ref{ex:2} and \ref{ex:3}. This seems to be a substantially more difficult question, and object of future work. On the other hand, it can be easily checked that, in  Example \ref{ex:1}, {\em $\NN_0$ is a $\WR0$ realization of $\NN$}. This means that {\em for any choice of rate constants} for $\NN$, the resulting dynamics is that of a $\WR0$ network $\NN_0$. This makes conclusions on the dynamics robust to changing rate constants. In practice their exact values may not be known, and the system may be highly sensitive to changes in rate constants. We now define what it means to be $\WR0$-realizable.

\begin{definition}
Let $\NN$ be a reaction network. Then $\NN$ is said to be $\WR0$-realizable iff for every choice of rate constants $\kappa$, there exists a weakly reversible network $\NN'$ and a rate constant vector $\kappa'$ such that $(\NN',\kappa')$ is a realization of $(\NN,\kappa)$.
\end{definition}

The following question arises naturally: \\

{\bf Question 2.}
{\em Given a reaction network $\NN$, is $\NN$ {\em $\WR0$-realizable}, i.e. does the mass-action ODE system of $(\NN,\kappa)$ have a $\WR0$ realization for any choice of rate constants $\kappa$}?\\

We provide an answer in the algorithm we introduce next. We need a few definitions/notations for the algorithm and its proof of correctness. 

\begin{definition}
Let $\cal N$ be a $\WR0$ network and $y$ be a source complex of $\NN$. Then

\begin{enumerate}[label=(\roman*)]

\item ${\cal L}_{\cal N}(y)=Cone\{y'-y\ |\ y'\, \text{is a complex of }\NN \text{ in the same linkage class as } y\}$. Note that $Cone_{\NN}(y)\subseteq {\cal L}_{\cal N}(y)$.

\item $\Phi_{{\cal L}_{\NN}}(w_y)$ denotes the face of ${\cal L}_{{\cal N}}(y)$ generated by column $w_y$ of $W$. $\Phi_{{\cal L}_{\NN}}(w_y)$ determines which reactions $y\to y_i$ belong to $\NN$ (lines 13-20 in Algorithm \ref{alg:polly}). In particular, for any $1\le i\le m$ with $y_i\neq y$, we have $y\to y_i\in {\cal N}$ if and only if $y_i-y\in \Phi_{{\cal L}_{\NN}}(w_y)$.

\end{enumerate}

\end{definition}

\subsection{Main results}

Let $\NN$ be a reaction network. If the ODE system of $(\NN,\kappa)$ has $\WR0$ realization, then this is denoted by $(\NN_\kappa, k_\kappa)$. 

\medskip

\begin{breakablealgorithm}
\label{alg:allk}
\caption{$\WR0$ realizations of reaction networks}

\medskip

\noindent \textbf{Input:}  A reaction network $\NN$.

\noindent \textbf{Output:} Either ``True" and return the $\WR0$ realization $\NN_0$ of $\NN$, or ``False". 

\begin{algorithmic}[1] 

\State Run Algorithm~\ref{alg:polly} on the mass-action system $(\NN,{\bf 1})$ (all rate constants are set to 1)

\If{Step 1 outputs ``Fail"}

\State Print: False. There exist rate constants $\kappa$ such that the mass-action ODE system of $(\NN,\kappa)$ does not have $\WR0$ realization. exit.

\Else
\State Return the underlying network $\NN_{\bf 1}$ of the $\WR0$ realization from Step 1. 

\For{every for source complex $y$ of $\NN$} 

\If{$Cone_\NN(y)\not\subseteq Cone_{\NN_{\bf 1}}(y)$ }

\State Print: False. There exist rate constants $\kappa$ such that the mass-action ODE system of $(\NN,\kappa)$ does not have $\WR0$ realization. exit.

\EndIf
\EndFor
\EndIf

\State  $\NN$ has a $\WR0$ realization $\NN_{\bf{1}}$ from Step 5.

\end{algorithmic}

\end{breakablealgorithm}
\medskip

\medskip

Clearly, if $\NN$ has a $\WR0$ realization, then it is $\WR0$-realizable. 


\begin{proof}[Proof of Algorithm 2]

If line 1 outputs ``Fail" then $(\NN,{\bf 1})$ does not have $\WR0$ realization, so $(\NN,\kappa)$ does not have $\WR0$ realization for all $\kappa$.
If line 1 outputs ``Success" then lines 6-11 check that the dynamics of $\NN$ is included in that of $\NN_{\bf{1}}$, in other words, for any rate constants $\kappa$ of $\NN$ there exist rate constants $\cal K$ of $\NN_{\bf{1}}$ such that the two mass-action systems are the same. This is equivalent to checking that 
$Cone_\NN(y)\subseteq Cone_{\NN_\kappa}(y)$
for all source complexes $y$ of $\NN$ by Theorem \ref{thm:dynIncl}. If this is true, then $(\NN,\kappa)$ has $\WR0$ realization for all $\kappa$. If, on the other hand, the algorithm goes through line 8, then there are rate constants $\kappa_0$ such that $(\NN,\kappa_0)$ does not have a $\WR0$ realization with underlying network $\NN_{\bf 1}$. Theorem \ref{thm:main} then implies that $(\NN,\kappa)$ does not have $\WR0$ realization for all $\kappa$.
\end{proof}

We now state our main result in Theorem~\ref{thm:main} which plays a key role in the correctness of Algorithm \ref{alg:allk}.

\begin{theorem}\label{thm:main} 
Let $\NN$ be a $\WR0$-realizable network. Let $(\NN_\kappa, k_\kappa)$ be a $\WR0$ realization of $(\NN,\kappa)$. Then $\NN_\kappa$ does not depend upon $\kappa$.
In other words, if $\NN$ is a $\WR0$-realizable network, then there exists a unique $\WR0$ network $\NN'$ such that $\NN$ is realizable by $\NN'$. 
\end{theorem}

In what follows, we will prove certain lemmas that will be used in the proof of Theorem~\ref{thm:main}.
Throughout this exposition, $\NN$ will denote a $\WR0$-realizable network. For any rate constant vector $\kappa$ for $\NN$ we write $f_{(\NN,\kappa)}(x)=W_\kappa x^{Y_\kappa}$
where $W_\kappa$ has no zero columns. Recall that the $\WR0$ realization of $(\NN,\kappa)$ is denoted by $(\NN_{\kappa},k_{\kappa})$.

The next lemma shows that if $\NN$ is $\WR0$-realizable then for any source complex $y$ of $\NN$ and any rate constants $\kappa$, the monomial $x^{y}$ appears in $f_{(\NN,\kappa)}(x)$, i.e. the net reaction vector of any source  complex $y$ of $\NN$ is nonzero for all rate constants $\kappa$ of $\NN$. 

\begin{lemma}\label{lem:monPersists}
 If $\NN$ is $\WR0$-realizable then for any source complex $y$ of $\NN$ we have $0\notin \mathrm{int}(Cone_\NN(y))$.
\end{lemma}
\begin{proof}
For contradiction, assume that there exists a source complex $y_0$ of $\NN$ such that $0\in \mathrm{int}(Cone_\NN(y))$, i.e.,
$$\sum_{y_0\to y'\in\NN}\alpha_{y_0\to y'}(y'-y_0)=0$$ 
for some positive $\alpha_{y_0\to y'}$.
Let us fix a set of rate constants $\kappa^1$ for $\NN$ such that the net reaction vectors corresponding to all source complexes of $\NN$ are nonzero (i.e., there is no cancellation of monomials in $f_{(\NN,\kappa)})$. Let $\kappa^0$ be a new set of rate constants for $\NN$ obtained by replacing $\kappa^1_{y_0\to y'}$ by $\alpha_{y_0\to y'}$ for all rate constants of reactions starting at $y_0$. This implies that only the vertex $y_0$ has zero net reaction vector in $(\NN,\kappa)$.

Let $m$ denote the number of source complexes of $\NN$. Then $W_{\kappa^1}$ has $m$ columns and we may shuffle complexes so that $y_0$ is the last column of $Y_{\kappa^1}$, i.e. 
$W_{\kappa^1}=[W_{\kappa^0}|w_{y_m}]$.  
Let $\{c_1,\ldots,c_l\}\subset\R^m$ be a complete set of extreme rays of $\ker W_{\kappa^1}\cap\R^m_{\ge0}$. Then the first $m-1$ coordinates of these vectors 
are a complete set of extreme rays of $\ker W_{\kappa^0}\cap\R^{m-1}_{\ge0}$. Indeed, let $d_i\in {\mathbb R}^{m-1}$ denote the projection of $c_i$ in its first $m-1$ coordinates. For any $u\in\ker W_{\kappa^0}\cap\R^{m-1}_{\ge0}$ let $v=[u, 0]\in {\mathbb R}^m$. Then $v\in\ker W_{\kappa^1}\cap\R^{m}_{\ge0}$, so $v=\sum_{i=1}^l \alpha_i c_i$ for some nonnegative $\alpha_1,\ldots, \alpha_l$. Discarding the last coordinate gives 
$u=\sum_{i=1}^l \alpha_i d_i$. 
Finally, the supports of $c_1,\ldots c_l$ partition  $\{1,\ldots, m\}$ (Proposition \ref{prop:2}), so the supports of $d_1,\ldots d_l$ partition $\{1,\ldots, m-1\}$. Therefore $\{d_1,\ldots, d_l\}$ is a complete set of extreme rays of $\ker W_{\kappa^1}\cap\R^{m-1}_{\ge0}$.

Suppose $y_0, y_1,\ldots, y_t$ is the list of complexes in the linkage class of $\NN_{\kappa^1}$ containing $y_0$. Since  $c_i$ and $d_i$ have the same support (they are equal on the first $m-1$ coordinates), Proposition \ref{prop:2} implies that 
$\{y_1,\ldots, y_t\}$ is the set of complexes in some linkage class of $\NN_{\kappa^0}$. For any $i\in\{1,\ldots, t\}$ column $w_i=w_{y_i}$ of $W_{\kappa^0}$  belongs to ${\cal L}_{\NN_{\kappa^0}}(y_i)=Cone\{y_j-y_i|1\le j\le t, j\neq i\}$, which is a face of ${\cal L}_{\NN_{\kappa^1}}(y_i)=Cone\{y_j-y_i|0\le j\le t, j\neq i\}$ not containing $y_i-y_0$.
Therefore $y_i-y_0\notin \Phi_{  \cal{L}_{{\NN_{\kappa}^{\text{1}}}}}(w_i)$. This implies that $y_i\to y_0\notin {\NN_0}$ and hence none of the reactions $y_1\to y_0,\ldots, y_t\to y_0$ belong to $\NN_0$. This is a contradiction since $\NN_0$ is weakly reversible.

\end{proof}

\medskip

Next we show that for a $\WR0$-realizable network $\NN$, changing rate constants  constants does not change the number of linkage classes and the complexes in each linkage class in the $\WR0$ realization of $(\NN,\kappa)$. We formalize this using the idea of linkage class partition.

\begin{definition} The partition on the set of complexes of reaction network $\NN$ based on their linkage class membership is called
the {\em linkage class partition} of $\NN$.
\end{definition}

\begin{lemma}\label{lem:samelc}
Let $\NN$ be a reaction network. Let $(\NN_\kappa, k_\kappa)$ denote the set of $\WR0$ realizations of $(\NN,\kappa)$. Then $\NN_{\kappa}$ has the same linkage class partition for every choice of rate constants $\kappa$ of $\NN$.
\end{lemma}

\begin{proof}
Let $\kappa$ be an arbitrary vector of rate constants of $\NN$. Let $l$ be the number of linkage classes in $\NN_{\kappa}$. Fix $\kappa^0$ such that $\NN_{\kappa^0}$ is the network that has the maximum number $l_0$ of linkage classes among all $\NN_{\kappa}$.

Let $L$ be a linkage class of $\NN_{\kappa^0}$ and let $\kappa^1$ be obtained from $\kappa^0$ by replacing 
$\kappa^0_{y\to y'}$ by $\kappa_{y\to y'}$ for all reactions $y\to y'$ with source complex in $L$. Let $\{c_1,\ldots,c_{l_0}\}$ denote a complete set of extreme rays of $\ker W_{\kappa^0}\cap\R^m_{\ge0}$ and assume without loss of generality that complexes of $L$ correspond to the support of $c_{l_0}$ (using Proposition \ref{prop:2}).

Lemma~\ref{lem:monPersists} ensures that for all $\kappa$, $W_{\kappa}$ has $m$ columns, corresponding to the $m$ source monomials of $\NN$. For any source complex $y$ of $\NN$ not belonging to $L$, the corresponding columns of $W_{\kappa^0}$ and $W_{\kappa^1}$ are equal. The remaining columns correspond to complexes in $L$, which in turn correspond to zero entries in $c_j$, for $j<l_0$ since those $c_j$ are supported on different linkage classes. For $1\le j<l_0$, we have  $W_{\kappa^1}c_j=W_{\kappa^0}c_j=0$, and 
hence $c_j\in\ker W_{\kappa^1}\cap\R^m_{\ge0}$. 

Suppose $\NN_{\kappa^1}$ has $l_1$ linkage classes and let $d_1,\ldots, d_{l_1}$ be a complete set of extreme rays of $W_{\kappa^1}\cap\R^m_{\ge0}$. This is also a basis of $\ker W_{\kappa^1}$ and we can write $c_j$ as a positive linear combination of vectors $\{d_i|i\in I_j\}$, where the index sets $I_j$ are determined uniquely by $j$. For $i\in I_j$ the support of $d_i$ is included in that of $c_j$, and since $c_1,\ldots, c_{l_0}$ have pairwise disjoint supports, it follows that the index sets $I_1,\ldots, I_{l_0-1}$  are pairwise disjoint. Therefore 

\begin{equation}\label{eq:ineq}
l_1\ge \sum_{i=1}^{l_0-1}|I_j|+1\ge l_0
\end{equation}
where the term 1 accounts for at least a linkage class of $\NN_{\kappa^1}$ containing the complexes in $L$. But $l_0$ is the maximum possible number of linkage classes, so $l_1=l_0$. 
Therefore we have equality in (\ref{eq:ineq}) and  $|I_j|=1$ for $1\le j\le l_0-1$, in other words $c_j$ is an extreme ray of $\ker W_{\kappa^1}\cap\R^m_{\ge0}$. Moreover, there is a linkage class of $\NN_{\kappa^1}$ with the same complexes as $L$ and so $c_{l_0}$ is an extreme ray of $\NN_{\kappa^1}$ as well. This shows that $\NN_{\kappa^0}$ and $\NN_{\kappa^1}$ have the same linkage class partition.

Since $\NN_{\kappa^1}$ has maximum number of linkage classes, we may repeat the procedure above: pick a linkage class of $\NN_{\kappa^1}$ that contains at least one source complex whose corresponding rate constants have not been updated from $\kappa^0$, and replace $\kappa^0_{y\to y'}$ with $\kappa_{y\to y'}$ for all source complexes $y$ in that class. At each step the rate constants corresponding to at least one source complex get updated from $\kappa^0$ to $\kappa$, and the procedure terminates, $\NN_s=\NN$. We conclude that $\NN_{\kappa^0}$ and $\NN_{\kappa}$ have the same linkage class partition. 
\end{proof}

We now have all the ingredients for the proof of Theorem~\ref{thm:main}.

\begin{proof}[Proof of Theorem \ref{thm:main}]
Let $\NN$ be $\WR0$-realizable. By Lemma \ref{lem:monPersists} $\NN_\kappa$ has the same complexes (source complexes of $\NN$). Let $y$ be a source complex of $\NN$ and $\{y, y_1,\ldots, y_t\}$ denote its class in the linkage class partition in any $\WR0$ realization. By Lemma~\ref{lem:samelc}, $\NN_\kappa$ has the same linkage class partition for all $\kappa >0$. Then ${\cal L}_{\NN_{\kappa}}(y)=Cone(y_i-y|1\le i\le t)$ does not depend on $\kappa$. Therefore, we will denote ${\cal L}(y) ={\cal L}_{\NN_{\kappa}}(y)$ to indicate that ${\cal L}_{\NN_{\kappa}}(y)$ does not depend on $\kappa$.\\

\noindent Note that since $\kappa>0$, $w_y$ lies in the relative interior of $Cone_{\NN}(y)$. Therefore, we get 
\begin{eqnarray}
\Phi_{{\cal L}(y)}(w_y)=\Phi_{{\cal L}(y)}(Cone_{\NN}(y))
\end{eqnarray}
This implies that for any choice of $\kappa$, the face of ${\cal L}(y)$ generated by column $w_y$ corresponding to $y$ in $W_\kappa$ does not depend on $\kappa$. As a consequence
, all networks $\NN_{\kappa}$ have the same set of reactions with source complex $y$. But since $y$ was arbitrary, $\NN$ is realizable by the same $\WR0$ network.

\end{proof}


\section{Software implementation}\label{sec:soft}
Algorithm \ref{alg:allk} has been implemented in CoNtRol, an open source framework for the analysis of chemical reaction networks, commonly used by researchers in the field. CoNtRol has multiple tests, each analyzing the structure of a given reaction network to check if it has a given dynamical property. CoNtRol is hosted at West Virginia University, \href{https://control.math.wvu.edu}{https://control.math.wvu.edu}.

Algorithm \ref{alg:allk} relies on Algorithm \ref{alg:polly} which is implemented numerically. In Algorithm \ref{alg:allk}, We have set a threshold such that any value $|x|\le10^{-13}$ is considered to be zero. Algorithm \ref{alg:polly} uses GNU GLPK which ships with octave version 6.4.0 \cite{octave} to find a minimal set of extreme rays of $\ker(W)\cap\R^m_{\ge0}$. Based on the precision of the GNU GLPK Algorithm, we have set Algorithm \ref{alg:polly} to have a threshold such that any value $|x|\le 10^{-7}$ is considered to be zero.

\section{Discussion and Future Work}\label{sec:future}

Weakly reversible deficiency zero networks (abbreviated as $\WR0$ networks) are known to display robust dynamics. These networks are complex balanced for every choice of rate constants. For $\WR0$ networks, there exists a Lyapunov function (called the Horn-Jackson Lyapunov function), which implies that their equilibria are locally asymptotically stable~\cite{horn1972general}.Further, they are conjectured to be even globally stable~\cite{craciun2015toric}.  

It is instructive to analyze $\WR0$ networks in the context of a property called \emph{dynamical equivalence}.
Dynamical equivalence is a  property by which two different reaction networks can generate the same dynamical system. This property has been exploited to analyze the dynamics of networks for which classical theorems in reaction network theory stand silent. A first step in this direction was taken in~\cite{craciunalgorithm}. In particular, this is the content of Question 1 in our paper: Does a given mass-action system have $\WR0$ realization? An algorithm (labeled as Algorithm~\ref{alg:polly}) was proposed that takes an input a dynamical system and outputs whether or not there exists a $\WR0$ realization.

Taking this forward, we answer a related question (phrased as Question 2 in this paper): Given a reaction network $\NN$, of $(\NN,\kappa)$, does the mass-action ODE system of $(\NN,\kappa)$ have a $\WR0$ realization for any choice of rate constants $\kappa$? We show in Theorem~\ref{thm:main} that if such a $\WR0$ realization exists, then it is unique. Further, this is used to provide an algorithm (labelled as Algorithm~\ref{alg:allk}) for finding this realization when it exists.

There is a a larger question naturally related to  Questions 1 and 2 above, which we are currently exploring.

\medskip
{\bf Question 3.} Given a reaction network $\NN$, what is the set of rate constants $\kappa$ such that the mass-action system $(\NN,\kappa)$ has a $\WR0$ realization? 

\medskip

This seems to be a hard question, which involves symbolic convex cone computation and may require tools from real algebraic geometry. We intend to start with simple classes of such networks, like those in examples~\ref{ex:1} and \ref{ex:2} where these parameter regions were determined by inspection. The interesting connection between simplicial geometry and $\WR0$ networks which we started to explore will be useful tool in answering this question, particularly when the rate constants are constrained to polyhedral cones. 

\bibliographystyle{unsrt}
\bibliography{Bibliography}

\end{document}